\documentclass[conference]{IEEEtran}
\IEEEoverridecommandlockouts
\usepackage{cite}
\usepackage{amsmath,amssymb,amsfonts}
\usepackage{graphicx}
\usepackage{textcomp}
\usepackage{xcolor}
\usepackage{comment}

\usepackage{algorithm}
\usepackage{algpseudocode}
\usepackage{graphicx}
\usepackage{tikz}

\def\BibTeX{{\rm B\kern-.05em{\sc i\kern-.025em b}\kern-.08em
    T\kern-.1667em\lower.7ex\hbox{E}\kern-.125emX}}

\newboolean{arxivversion} 
\setboolean{arxivversion}{true}

\ifarxivversion
  \newenvironment{onlyarxiv}{}{}
\else
  \excludecomment{onlyarxiv}
\fi

\ifarxivversion
  \excludecomment{onlyisncc}
\else
  \newenvironment{onlyisncc}{}{}
\fi

\begin{document}

\newcommand{\M}{\mathcal{M}}
\newcommand{\X}{\mathcal{X}}
\newcommand{\R}{\mathbb{R}}
\newcommand{\deff}{:=}
\newcommand{\vertt}[1]{\vert #1 \vert}
\newcommand{\Separatable}[2]{(#1, #2)\text{-separable}}
\newcommand{\Separator}[1]{#1\text{-separator}}

\newtheorem{de}{Definition}
\newtheorem{lem}{Lemma}
\newtheorem{pro}{Proposition}
\newtheorem{example}{Example}
\newtheorem{theo}{Theorem}
\newenvironment{proof}
	{\begin{trivlist}\item[\hskip\labelsep{\em Proof.}]}
	{\leavevmode\unskip\nobreak\quad\hspace*{\fill}{\ensuremath{{\square}}}\end{trivlist}}

\algtext*{EndIf}
\algtext*{EndFor}

\algrenewcommand\algorithmicrequire{\textbf{Input:}}
\algrenewcommand\algorithmicensure{\textbf{Output:}}

\title{Differentially Private All-Pairs Shortest Distances for Low Tree-Width Graphs
\ifarxivversion
\else
\thanks{979-8-3503-3559-0/23/\$31.00 ©2023 IEEE}
\fi
}

\author{
\IEEEauthorblockN{1\textsuperscript{st} Javad B. Ebrahimi}
\IEEEauthorblockA{\textit{Sharif University of Technology} \\
\textit{Institute for Research in Fundamental Sciences (IPM)}\\
Tehran, Iran \\
javad.ebrahimi@sharif.edu}
\and
\IEEEauthorblockN{2\textsuperscript{nd} Alireza Tofighi Mohammadi}
\IEEEauthorblockA{\textit{Sharif University of Technology} \\
Tehran, Iran \\
a.tofighi77@sharif.edu}
\and
\IEEEauthorblockN{3\textsuperscript{rd} Fatemeh Kermani}
\IEEEauthorblockA{\textit{Sharif University of Technology} \\
Tehran, Iran \\
f.kermani@sharif.edu}
}

\maketitle
\ifarxivversion
    \thispagestyle{plain}
    \pagestyle{plain}
\fi

\begin{abstract}
In this paper, we present a polynomial time algorithm for the problem of differentially private all pair shortest distances over the class of low tree-width graphs. Our result generalizes the result of Sealfon~\cite{Shortest-Paths-and-Distances-with-Differential-Privacy-Sealfon} for the case of trees to a much larger family of graphs. Furthermore, if we restrict to the class of low tree-width graphs, the additive error of our algorithm is significantly smaller than that of the best known algorithm for this problem, proposed by Chen et. al. in \cite{chen2023differentially}. 
\end{abstract}


\section{Introduction}

\subsection{Differential Privacy}
Privacy-preserving data analysis 
is a way of learning about population 
while keeping confidential information about individuals private.
\textit{Differential privacy} introduced in the work of Dwork et. al. \cite{Calibrating-Noise-to-Sensitivity-in-Private-Data-Analysis-Dwork-etal} is a definition and clarification of this concept.

There are two main ingredients to make an algorithm privacy preserving. The first one is to answer queries in a randomized way.
That is, to output a
random element from the set of all possible outcomes. Equivalently, the output can be modeled as a probability distribution over the set of possible outcomes.
The second point is to make sure that adding or removing any individual from the dataset does not significantly change the output distribution. 
Datasets which are only different in one individual, are called neighboring datasets.
To generalize \textit{neighboring} definition in a more abstract context, we can consider the case where the domain of the algorithm is a metric space $(\X, d)$. In this case, we say
$x, y \in \X$ are neighbors if $d(x, y) \leq 1$.
As an example, 
when $\X$ is the set of integral vectors and $d$ is the $\ell_1$-metric, neighboring elements are of the form
$x, y \in \X$ 
such that $x, y$ differ in exactly one co-ordinate and in that co-ordinate, they are equal to two consecutive numbers.

One may observe that if a mechanism outputs the same distribution regardless of the input, it will be completely private. However, the main challenge in differential privacy is to design differential private mechanism with high level of privacy, while, approximating a target function reasonably accurate.
More precisely, for a target function $f$ on a metric space $(\X, d)$,  $\M$ is a DP mechanism if for any close point $x_1,x_2 \in \X$, 
$\M(x_1),\M(x_2)$ are \textit{close} distributions and 
$\M(x)$ is a good approximation of $f(x)$. In section \ref{preliminaries}
we make this definition mathematically precise. 

\subsection{Sensitivity}




The $\ell_1$ sensitivity of a function $f$
is the maximum value that the change of a single individual's data in input can change the function $f$'s output. Therefore, intuitively, sensitivity measures the minimum amount of uncertainty we must have in the response
in order to guarantee keeping an individual's data private.

The output of functions with higher sensitivity must be perturbed more to preserve a certain extent of privacy, and for functions with lower sensitivity, we can add lower noise to guarantee differential privacy with the same amount of privacy leak.

By choosing the noise from a family of the distributions, called Laplace distributions, we can guarantee that the mechanism is private and the error is small. For more details about Laplace noise and the proof of this statement, see
\cite{The-algorithmic-foundation-of-dp}
and for the formal definitions of the 
$\varepsilon$-differential privacy, error and Laplace mechanism see 
section
\ref{preliminaries}.



\subsection{Differential privacy on graphs}
One of the canonical problems in graph theory and computer science is to compute the shortest path between a pair of vertices of a given graph.
This problem is well-explored and efficient algorithms have been proposed in the literature. However, we are interested in the differential private version of this problem. 
This problem was first proposed by
Sealfon in 
\cite{Shortest-Paths-and-Distances-with-Differential-Privacy-Sealfon}.
Now, we explain the model of the problem.

Let 
$G = (V, E)$ be a graph on the vertex set $V$ and the edge set $E$ and let $w$ be the edge weight function that assigns non-negative weights to the edges of the underlying graph 
$G = (V, E)$. 
Suppose that $G$ is publicly known but $w$ is private. 
 Each dataset is a weight function $w$ and two 
datasets $w_1, w_2$ are neighbors when 
$\Vert w_1 - w_2\Vert_1 \deff \sum_{e \in E}\vertt{w_1(e) - w_2(e)} \leq 1$.


In the all pairs shortest distances problem, the target function is $f: \X \to \R^{\binom{n}{2}}$. Where $\X$ is the set of all possible weights on edges of $G$ 
and for a $w \in \X$, $f(w)$ is a vector of shortest distances of all pairs of vertices.

The goal is to efficiently construct a differential private mechanism to approximate $f$ with minimum $\ell_\infty$ error. We call this error the additive error of the algorithm.

The rest of this paper is organized as follows. In the next section, we review the history of differentally private APSD problem and the previous results on this problem. In Section \ref{section:technical-overview}, we try to explain the main challenges to generalizing the Sealfon's approach from trees to low tree-width graphs. Then we describe our solution to overcome these challenges in several stages. Also, we explain our algorithm for differentially private APSD problem informally. A high level justification of the correctness of the algorithm is also provided in that section. Next, we overview some preliminaries and introduce basic definitions and notations. Section \ref{section:algorithm} consists of the formal description of the algorithm, the proof of its correctness and analysing the time complexity of the algorithm. Finally, in the last section we conclude our paper.  

\section{Related Works}

Differential privacy has been applied to graph problems, including the all-pairs shortest distances (APSD for short) problem. Sealfon \cite{Shortest-Paths-and-Distances-with-Differential-Privacy-Sealfon} was the first to formally study APSD with privacy. The paper introduces a model for differentially private analysis of the shortest distances in a weighted graphs in which the graph topology is assumed to be publicly known and the private information consists only of the edge weights. In the DP framework, he required that the algorithm be $(\varepsilon, \delta)$-DP, where neighboring data sets (i.e. inputs) correspond to those whose weight vectors $w,w'$ which differs by at most 1 in the $\ell_1$-distance. Sealfon gave an $O(n\log n/\varepsilon)$-error "input-perturbation" algorithm for APSD, which adds Laplace noise to all edge weights and computes the shortest path in this resulting graph with noisy weights. 

When the underlying graph is a tree, Sealfon developed a DP mechanism with significantly smaller error. His main idea is to employ a standard technique in graph algorithm known as ``Centroid decomposition" of trees. He first finds the unique path from the root to all the vertices and then, for every pair of vertices, with 3 queries, he can compute the distances between the two. 

The idea is to split the tree into subtrees of at most half the size of the original tree. As long as we can release the distance from the root to each subtree with small error, we can then, recurse on the subtrees. Sealfon showed that there is an algorithm that is $\varepsilon$-differentially private on $T$ such that on the input $w : E \to \R^+$, outputs approximate distances between all pairs of vertices. Also, with probability at least $1 - \gamma$, the additive error on each output distance is $O(\log^{2.5} V \cdot \log(1/\gamma)/\varepsilon)$ for any $\gamma \in (0, 1)$.

Fan and Li \cite{fan2022distances} revisited the problem of privately releasing approximate distances between all pairs of vertices in a graph. They proposed improved algorithms with smaller error term to that problem for grid graphs and trees.

Fan et al. \cite{fan2022breaking} also generalized Sealfon's approach for trees to graphs that with removing few nodes become acyclic. The subset of vertices that removing make graph without any cycles called feedback vertex set.

Chen et al. in \cite{chen2023differentially} also studied this problem. They gave an $\varepsilon$-DP algorithm with additive error $\tilde{O}(n^{2/3}/\varepsilon)$ and an $(\varepsilon, \delta)$-DP algorithm with additive error $\tilde{O}(\sqrt{n}/\varepsilon)$ where $n$ denotes the number of vertices. This is the best known additive error for arbitrary graphs.


\section{Technical Overview}
\label{section:technical-overview}

Two of the commonly used ideas of differential privacy, are to add random noise to the input (input perturbation), or to add random noise to the output of the algorithm (output perturbation).
In both cases, the output of the algorithm is a random function which estimates the desired function and has some error. 



For instance, if we add \textit{i.i.d.} Laplace noises to the input, (i.e. edge weights),
and then, compute the shortest path according to noisy weights,
the magnitude of error to achieve $\varepsilon$-DP is $1/\varepsilon$.
\footnote{In order to achieve $\varepsilon$-DP, we can add a Laplace noise with parameter $1/\varepsilon$. In Section \ref{preliminaries} we define the Laplace distribution and see that the magnitude of such noise will be proportional to $1/\varepsilon$ with high probability.}
Since every path in a graph has at most $n-1$ edges, we may need to use $\Theta(n)$ noisy values to compute each shortest distance. Thus, the error is proportional to $\Theta(n \cdot 1/\varepsilon)$ (see Algorithm 3 of \cite{Shortest-Paths-and-Distances-with-Differential-Privacy-Sealfon}).
Alternatively, if we add Laplace noise to the output, i.e. the weights are exact and after computing the shortest distances we add the noise to the result. Since there are $\Theta(n^2)$ pairs of shortest distances, in order to achieve $\varepsilon$-DP, the magnitude of the noise must be $\Theta(n^2/\varepsilon)$. Thus, the error is proportional to $\Theta(n^2/\varepsilon)$ (see Section 4 of \cite{Shortest-Paths-and-Distances-with-Differential-Privacy-Sealfon}).

Note that in the both cases, the error is proportional to the 
product of the magnitude of noise and the number of needed noisy values to release each shortest distance. Therefore, roughly speaking, in order to reduce the magnitude of error, we must reduce either the magnitude of the noises or the number of added noises or both.  

For trees, Sealfon in Section 4.1 of \cite{Shortest-Paths-and-Distances-with-Differential-Privacy-Sealfon} invented an elegant idea to decrease the magnitude of noise and the number of noisy values needed
to compute each of the shortest distances.

The algorithm constructs an intermediate graph,
which has two important properties: First, if we change the weights of the input graph by at most 1, only $O(\log n)$ of the weights in the intermediate graph will change and they change no more than 1. Secondly, for all $v, u \in V$, the weight of an $O(\log n)$-hop shortest path between $v$ and $u$ in the intermediate graph, is equal to the shortest distance between $v$ and $u$ in the input graph. Thus, we can use $O(\log n / \varepsilon)$-magnitude Laplace noises and we only need $O(\log n)$ noisy values to compute each shortest distance. Hence, the magnitude of the error is roughly $O(\log^2 n/\varepsilon)$.
\footnote{The exact value is $O(\log^{2.5} n/\varepsilon)$. The extra log factor is because this property with high probability.}

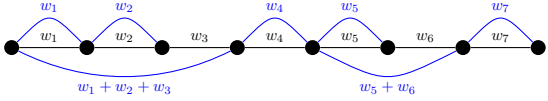
\begin{figure}[htbp]
\label{fig1:example}
\begin{center}
\begin{tikzpicture}
  [scale=1,auto=center]
  \tikzset{graphnode/.style={circle,fill=black!100,minimum size = 2, scale = .4, inner sep=5pt}}
  \tikzset{weightnode/.style={scale = .6}}
  \node[graphnode] (n1) at (0,0)  {};
  \node[graphnode] (n2) at (1,0)  {};
  \node[graphnode] (n3) at (2,0)  {};
  \node[graphnode] (n4) at (3,0)  {};
  \node[graphnode] (n5) at (4,0)  {};
  \node[graphnode] (n6) at (5,0)  {};
  \node[graphnode] (n7) at (6,0)  {};
  \node[graphnode] (n8) at (7,0)  {};
  
  \foreach \from/\to/\weight in {n1/n2/w_1,n2/n3/w_2,n3/n4/w_3,n4/n5/w_4,n5/n6/w_5,n6/n7/w_6,n7/n8/w_7}
    \draw (\from) -- node[weightnode,above] {$\weight$} (\to) ; 
  
  \foreach \x/\y in {1/2,2/3,4/5,5/6,7/8}
  \draw[blue] (n\x) .. controls (\x - 0.5,0.5) ..node[weightnode, above] {$w_\x$}(n\y);
  \draw[blue] (n1) .. controls (1, - 0.5) and (2, -0.5) ..node[weightnode, below] {$w_1 + w_2 + w_3$}(n4);
  \draw[blue] (n5) .. controls (5, - 0.5) ..node[weightnode, below] {$w_5 + w_6$}(n7);

\end{tikzpicture}
\caption{Making an intermediate graph for a path. black edges are the base graph and the blue edges are the intermediate graph.}
\label{figure:sealfonExample}
\end{center}

\end{figure}

\begin{example}[Sealfon's construction of the intermediate graph for a path ]
Figure \ref{figure:sealfonExample} gives
an example of the graph made by Sealfon's algorithm.
The $\ell_1$ sensitivity of the graph is at most 2, i.e. changing the weight of a black edge affects at most 2 of the blue edges
\end{example}

In this paper, we aim to extend the Sealfon's idea beyond trees. We address several challenges due to inherent differences between trees and low tree-width graphs. These challenges include:
\begin{enumerate}
    \item The edges of the shortest path between two specific vertices in a tree, only depend on the publicly known topology of the graph and they are unique. 
    So the desired algorithm, only needs to privately release the summation of private weights of the edges of the unique paths.
    Clearly, this property is not true for non-tree graphs. \vspace{-3mm} \label{challenge:unique}
    \item In trees, due to uniqueness of shortest paths, it is possible to calculate the all pairs shortest distances using only the shortest distances between the root and other vertices. This can be achieved by computing the lowest common ancestor for each pairs of vertices and calculating the distances from the root to every vertex.
    This property of trees allows for efficient calculation of all pairs shortest distances using only the root to all distances. \vspace{1mm} \label{challenge:root-to-all}
    \item 
    Trees have a vertex called the centroid. The centroid is a vertex that, removing it, separates the tree into connected components which have sizes no greater than half the size of the original tree.
    \label{challenge:centroid}
\end{enumerate}


\vspace{0.2cm}
This paper introduces Algorithm \ref{alg:apsd-low-tree-width-main} that calculates all pairs shortest distances with differential privacy on low tree-width graphs and overcome the challenges above.
The algorithm takes $G=(V, E)$ and $T$ as input.
Where $T$ is tree decomposition of $G$ and $T$ has width $p$. Also assume that $\vert V \vert = n$.
The algorithm consists of three stages. In the first stage, we construct an intermediate graph with the properties that mentioned above,
in the second stage, we ensure the differential privacy by adding noises to the intermediate graph and in the last one, the algorithm estimates all pairs shortest distances.
\\
\textbf{Stage 1:} Algorithm \ref{alg:cunstruct-shortcut-graph}, uses a divide and conquer approach that constructs an intermediate graph $G^\prime$ by adding edges to the original graph. By proposition \ref{prop:tree-width-implies-separator}, the algorithm exploits $T$ and finds a subset of vertices, called separator, which removing it partitions graph into components with half size. Thus, the depth of recursion call of the algorithm will be $O(\log n)$. We use separator instead of centroid to address challenge \ref{challenge:centroid}.

The edges of intermediate graph are computed by the \textproc{ComputeEdges} recursive function in Algorithm \ref{alg:compute-shortcut-edges}. This function takes as input, a weighted graph $G=(V, E)$, a tree decomposition $T$ of $G$, and a subset of vertices $V_0$ of $G$. The union of $V_0$ and separator act like root in tree. we call the union, $V^\prime_0$. We deal challenge \ref{challenge:root-to-all} using $V^\prime_0$. In each recursion call of \textproc{ComputeEdges}, lemma \ref{lemma:FewToAllShortcuts-Length-on-v0S} shows that  the shortest distances between all $v \in V^\prime_0$ and $V$ is equals to $O(\log n)$-hop shortest distance between $v$ and $u$ in the returned graph.
\\
We need to extend the result to all pairs shortest distances. To do so, we need a wise selection of smaller problems and $V_0$. We construct smaller versions of the problem by combining the connected components after removing the separator and the separator itself. We also propagate $V^\prime_0$ for $V_0$ of smaller versions of the problem in the divide and conquer approach. In Lemma \ref{lemma:Gprime-and-G-APSD-is-same}, we show that with this selection, we can extend the result from $V_0^\prime$ to $V$ for all pairs shortest distances.
\\
\textbf{Stage 2:} Lemma \ref{lemma:sensitiviy-of-algorithm-is-low} shows that the sensitivity of the intermediate graph is low. Thus, the algorithm then adds Laplace noise to the edge weights of $G^\prime$ to provide differential privacy. The scale of the Laplace noise is $O(p^2 \log^2 n / \varepsilon)$, where $p$ is the tree-width of the graph.
\\
\textbf{Stage 3:} In the description of stage 1, we mentioned that for each $v, u \in V$, the shortest distance between $v$ and $u$ is equal to an $O(\log n)$-hop shortest distance in the intermediate graph. The algorithm uses post-processing Proposition \ref{Post-Processing-proposition} to return all-pairs shortest path distances for paths with at most $O(\log n)$ hops in the graph $G^\prime$ with noisy edge weights. The minimization of the weight of all $O(\log n)$-hop paths addresses challenge \ref{challenge:unique}.

We uses Laplace noises, thus, with high probability, the magnitude of each noise is proportional to $O(\log n \cdot p^2 \log^2 n / \varepsilon)$. We minimize over $O(\log n)$-hop paths, thus we use only $O(\log n)$ noises to estimate each shortest distance. Thus, with high probability, the magnitude of the additive error of each shortest distance is at most $O(p^2 \log^4 n / \varepsilon)$. This result outperforms previous results in $\varepsilon$-DP when $p = o(n^{1/3})$.

\begin{onlyarxiv}
In section \ref{section:appendix}, we provide a detailed description of the technical aspects of the algorithm and its proof. 
\end{onlyarxiv}
\begin{onlyisncc}
Proofs of the technical lemmas are omitted due to the space constraint but can be found in \cite{arxivversion}. 
\end{onlyisncc}

\section{Preliminaries}
\label{preliminaries}

In this section we describe the main concepts of graph theory and differential privacy, that are used in this paper.

\subsection{Graph Theory}
Let $G=(V, E)$ denote an undirected graph with vertex set $V$ and edge set $E$. we also show $V$ and $E$ by $V(G)$ and $E(G)$ respectively and let $w: E \rightarrow \R^{+}$ be a weight function. Let $\vert V\vert$ and $\vert E\vert$ be the number of vertices and edges, respectively.
For $X, Y \subseteq V$, let donate $E[X, Y]$ be all edges between $X$ and $Y$.

Let $\mathcal{P}^G_{x y}$ denote the set of paths between a pair of vertices $x, y \in V(G)$. For any path $P \in \mathcal{P}^G_{x y}$, the weight $w(P)$ is the sum $\sum_{e \in P} w(e)$ of the weights of the edges of $P$. The distance $d_{G, w}(x, y)$ from $x$ to $y$ denotes the weighted shortest distance $\min_{P \in \mathcal{P}^G_{x y}} w(P)$. We will denote the hop length $\ell(P)$ of path $P=\left(v_0, \ldots, v_{\ell}\right)$ is the number $\ell$ of edges on the path. 
Also the $k$-hop distance $d_{G, w}^{k}(x, y)$ between $x$ and $y$ denotes as the minimum weight between at most $k$-hop paths $d_{G, w}^k = \min _{P \in \mathcal{P}^G_{x y}, \ell(P) \leq k} w(P)$.




\begin{de}[Tree decomposition]
\cite{diestel2010graph,bodlaender1998tourist}
A tree decomposition of a graph $G(V, E)$ is a labeled tree $T$, where each node $i$ of $T$ is labeled by a subset (bag) $B_i \subset V$ of vertices of $G$, each edge of $G$ is in a subgraph induced by at least one of the $B_i$, and the nodes of $T$ labeled by any vertex $v \in V$ are connected in $T$. 
The width of $T$ is maximum of cardinality of bags of $T$ minus 1.

\end{de}

\begin{de}[Tree-width]
\cite{diestel2010graph,bodlaender1998tourist}
The tree-width of $G$ is the minimum integer $p$ such that there exists a tree decomposition $G$ with width $p$.
\end{de}
It is known that a tree has tree-width 1, a series-parallel graph has tree-width 2, a $k$-clique has tree-width $k-1$, and an $n$ by $n$ grid has tree-width $\Theta(n)$.

\begin{pro}[Reduced tree decomposition]
\label{prop:tree-decomposition-reduction}
Let $G = (V, E)$ be a graph of tree-width $k$, and let $H = (V', E')$ be a subgraph of $G$. If $T$ a width-$k$ tree decomposition of $G$, then, by removing nodes that are not in $V'$ from the bags of $T$ and removing any bags that become subsets of other bags, we obtain a reduced tree decomposition $T'$ of $H$ with width $k$, and no bag of $T'$ being a subset of another.
\end{pro}




\begin{pro}[Tree-width implies separator] 
\label{prop:tree-width-implies-separator}
\cite{feige2012treewidth}
Let $G(V, E)$ be a graph of tree-width $p$.
Then, there is a set $S$ of at most $p+1$ such that every connected component of $G \setminus S$ has no more than half of vertices. We call this set the separator of graph and we can find it with $O(\vert V \vert^2)$ time complexity by tree decomposition of graph.
\end{pro}

\subsection{Differential Privacy}

We now formally define differential privacy in the private edge weight model.

\begin{de}[Neighboring weights] For any edge set $E$, two weight functions $w, w^{\prime}: E \rightarrow \R^{+}$ are neighboring, 
denoted $w \sim w^{\prime}$, if
$$
\left\|w-w^{\prime}\right\|_1=\sum_{e \in E}\left|w(e)-w^{\prime}(e)\right| \leq 1
$$
\end{de}

\begin{de}[Mechanism]
        A randomized algorithm $\M$
with domain $A$ and discrete range $B$ is associated with a mapping
$\M: A \to \Delta(B)$. On input $a \in A$, the algorithm $\M$ outputs $\M(a) = b$
with probability $(\M(a))_b$ for each $b \in B$. The probability space is over the coin flips of the algorithm $\M$.
\end{de}

\begin{de}[Differential privacy, \cite{Calibrating-Noise-to-Sensitivity-in-Private-Data-Analysis-Dwork-etal}]
        A randomized algorithm $\M: A \to R$ is $(\varepsilon, \delta)$-differentially private if for all $\mathcal{S} \subseteq R$ and for all $x, y \in R$ such that $x \sim y$ :
$$
\operatorname{Pr}[\M(x) \in \mathcal{S}] \leq \exp (\varepsilon) \operatorname{Pr}[\M(y) \in \mathcal{S}]+\delta
$$
where the probability space is over the coin flips of the mechanism $\M$.
\end{de}

If $\delta = 0$, we use $\epsilon$-differential privacy or $\epsilon$-DP as $\epsilon$-differential privacy.
    
\begin{de}[Sensitivity of a function]
    The $\ell_1$-sensitivity of a function $f: A \to \R^k$ is:
    $$\Delta_1 f = \max_{x,y \in A, x \sim y}  \Vert f(x) - f(y) \Vert_1 $$
\end{de}

The Laplace distribution with scale parameter $b > 0$ is defined by the probability density function
$$f(x) = \frac{1}{2b} \exp\left(-\frac{|x|}{b}\right).$$

For a random variable $X \sim \text{Lap}(b)$, the probability that $\vertt{X} \geq t\cdot b$ for some $t > 0$ is given by
$\Pr(\vertt{X} \geq t\cdot b) = 2 \exp(-t).$

\begin{de}[Laplace mechanism]
        Given any function $f : A \rightarrow \R^k$, the Laplace mechanism is defined as:
$$
\M_L(x, f(\cdot), \varepsilon)=f(x)+\left(Y_1, \ldots, Y_k\right)
$$
where $Y_i$ are \textit{i.i.d.} random variables drawn from $\text{Lap}(\Delta_1 f / \varepsilon)$.
\end{de}

In order to be able to use the Laplace mechanism, we should know that it preserve the privacy of the algorithm. Theorem 3.6 of \cite{The-algorithmic-foundation-of-dp} states as bellow.

\begin{pro}[Laplace mechanism is $\varepsilon$-DP \cite{The-algorithmic-foundation-of-dp}]
\label{Proposition:laplace-is-e-private}
The Laplace mechanism 
$\M_L(x, f(\cdot), \varepsilon)$
is $\varepsilon$-DP.
\end{pro}

\begin{pro}[Post-Processing] \label{Post-Processing-proposition}
Let $\M: \R^{\vertt{\X}} \rightarrow R$ be a randomized algorithm that is $(\varepsilon, \delta)$-differentially private. Let $f: R \rightarrow \R^{\prime}$ be an arbitrary randomized mapping. Then, $f \circ \M: \R^{\vertt{\X}} \rightarrow \R^{\prime}$ is $(\varepsilon, \delta)-$ differentially private.
\end{pro}

\section{Main result}
\label{section:algorithm}
In Section \ref{section:technical-overview}, we described the intuition behind the algorithm and its properties. In this section, we will present the details of the algorithm implementation in pseudo-code form. Next, we will present the technical aspects of its properties and correctness.

\begin{algorithm}[b!]
\caption{Determining Shortcut Edges for Addition to The Intermediate Graph}
\label{alg:compute-shortcut-edges}
\begin{algorithmic}[1]
    \Function{ComputeEdges}{$G=(V, E), w: E \to \R^+, T, V_0 \subseteq V$}
        \Require{$G$ is the graph with weight function $w$. $T$ is the tree decomposition of $G$ with maximum bag size $p+1$, and $V_0$ is a helper argument called the starting set.}
        \Ensure{A list of triples $(v, u, x)$ where $v, u \in V$ and $x \in \R^+$ that should be added to the intermediate graph for computing distances with few hops}
        \State $R \gets \emptyset$ \Comment{Array $R$ is a list of weighted edges}
        \If{$|V| \leq 6(p+1)$}
            \ForAll{$v \in V$}
                \ForAll{$u \in V$}
                    \State Insert $(v, u, d_{G, w}(v, u))$ into $R$
                \EndFor
            \EndFor
        \Else
            \Comment{Partition the graph into smaller components using a bag from the tree decomposition}
            \State\label{alg:line:use-s} Let $S$ be the bag of $T$ that partitions $G$ into components of size at most $|V|/2$ by proposition \ref{prop:tree-width-implies-separator}
            \State $V_0^\prime \gets V_0 \cup S$
            \State Let the connected components of $G\setminus S$ be $C_1, \ldots, C_l$
            
            \Comment{Add edges between vertices in the starting set and separator}
            \ForAll{$v \in V_0^\prime$}
                \ForAll{$u \in S$}
                    \State Insert $(v, u, d_{G, w}(v, u))$ into $R$
                \EndFor
            \EndFor
            
            \Comment{Recursively compute edges for each component}
            \ForAll{$i \in [l]$}
                \State\label{alg:line:def-ui} $H_i \gets (V(C_i) \cup S, E(C_i) \cup E[C_i, S])$
                \State $V_i \gets (V(C_i) \cap V_0^\prime)$
                \State Let $T_i$ be the reduced tree decomposition $T$ for $H_i$ using Proposition \ref{prop:tree-decomposition-reduction}
                \State Let $w_i$ be the restriction of function $w$ to edges in $H_i$
                \State $R_i \gets $\Call{ComputeEdges}{$H_i, w_i, T_i, V_i$}
                \ForAll{$(v, u, x) \in R_i$}
                    \State Insert $(v, u, x)$ into $R$
                \EndFor
            \EndFor
        \EndIf
        \State\Return $R$
    \EndFunction
\end{algorithmic}
\end{algorithm}
\begin{algorithm}[b!]
\caption{Constructing a Low $\ell_1$-Sensitivity Intermediate Graph for Efficient APSD Computation}
\label{alg:cunstruct-shortcut-graph}
\begin{algorithmic}[1]
	\Function{ConstructGraph}{$G=(V, E), w: E \to \R^+, T, V_0 \subseteq V$}
	\Require{Graph $G$ with weight function $w$, tree decomposition $T$ of $G$ with maximum bag size $p+1$, and starting set $V_0$}
	\Ensure{A graph in which all-pair shortest path distances of $G$ can be computed using only $O(\log \vert V \vert)$-hop paths}
    	\State $R \gets $ \Call{ComputeEdges}{$G, V_0$}
    	\State Create a copy of $G$ and call it $G^\prime$
    	\State Create a copy of $w$ and call it $w^\prime$
    	
    	\Comment{Add the computed shortcut edges to the new graph}
    	\ForAll{$(v, u, x) \in R$}
    	    \If {there is no edge between $v$ and $u$}
                \State Add edge between $v$ and $u$ in $G^\prime$
                \State Set $w^\prime(v, u) \gets x$
            \ElsIf{$x < w^\prime(v, u)$}
                \State Set $w^\prime(v, u) \gets x$
    	    \EndIf
    	\EndFor
    	\State \Return ($G^\prime$, $w^\prime$)
	\EndFunction
\end{algorithmic}
\end{algorithm}

    	
	
	
	

The following theorem, states that Algorithm \ref{alg:apsd-low-tree-width-main} has the desired properties. 

\begin{theo}[Main theorem]
\label{theorem:main}
Let $G = (V, E)$ be a graph of tree-width $p$. Then, there exists an $\varepsilon$-differentially private algorithm that takes as input a weight function $w: E \to \R^+$ and releases the all-pairs shortest path distances of $G$ with weight $w$. For any $\gamma \in (0, 1)$, the error is bounded by $O(\log(1/\gamma) \cdot p^2 \cdot \log^4(|V|)/\varepsilon)$ with probability at most $1-\gamma$. 
\end{theo}

To prove that Algorithm \ref{alg:apsd-low-tree-width-main} satisfies the conditions of our main theorem (Theorem \ref{theorem:main}), we need to establish several properties of Algorithms \ref{alg:compute-shortcut-edges}, \ref{alg:cunstruct-shortcut-graph}, and \ref{alg:apsd-low-tree-width-main}. 

We present these properties in the form of following lemmas. 
\begin{onlyarxiv}
The proofs can be found in Section \ref{section:appendix}.
\end{onlyarxiv}
\begin{onlyisncc}
The proofs can be found in appendix of \cite{arxivversion}.
\end{onlyisncc}

The first property we will show is that the recursion depth of Algorithm \ref{alg:compute-shortcut-edges} is not too large.
\begin{lem}[Recursion depth of algorithm is logarithmic]
\label{lemma:DepthOfAlgorithmisNotBig}
Let $G = (V, E)$ be a graph of tree-width $p$ and let $T$ be its corresponding tree-decomposition. Then, the recursion depth of Algorithm \ref{alg:compute-shortcut-edges} on input $G$ and $T$ is at most $O(\log \vert V \vert)$.
\end{lem}

In the previous lemma, we demonstrated that the depth of the Algorithm \ref{alg:compute-shortcut-edges} is $O(\log \vert V \vert)$. Since the size of the starting set $V_0$ increases by at most $O(p)$ at each depth, there is an upper bound for the size of the starting set in all recursive steps. The following lemma presents this condition in a formal format.

\begin{lem}[Upper bound on starting set size]
\label{lemma:SizeOfStartingSetsAreNotBig}
When running the $\textproc{ComputeEdges}(G, w, T, V_0)$ function on a graph $G = (V, E)$ with tree-width $p$ and a set of starting vertices $V_0 \subseteq V$, the size of the starting vertex set in all recursive steps of the function is bounded by $O(\vert V_0 \vert + p \cdot \log \vert V \vert)$.
\end{lem}

Now, it's time to find an upper bound for the sensitivity of the function.

\begin{lem}[Sensitivity of algorithm]
\label{lemma:sensitiviy-of-algorithm-is-low}
Let $G = (V, E)$ be a graph of tree-width $p$ and let $T$ be its corresponding tree-decomposition. Let $V_0 \subseteq V$ be an arbitrary set of starting vertices. If $f(w) = \textproc{ConstructGraph}(G, w, T, V_0)$, then, the $\ell_1$ sensitivity of $f$ is $O\left((|V_0| + p) \cdot p \cdot \log^2 |V| + p^2\right)$.
\end{lem}

By Lemma \ref{lemma:sensitiviy-of-algorithm-is-low}, we know that to guarantee $\varepsilon$-DP, the magnitude of each noise can be bounded by $O(p^2 \log^2 \vert V \vert/ \varepsilon)$. The next two lemmas state that we can also involve only a small number of shortcuts to compute the shortest distances. The following lemma states that we can compute the shortest distance between each pair $v \in V_0 \cup S$ and $u \in V$ using only $O(\log \vert V \vert)$-hop paths in the intermediate graph.

\begin{lem}[Distance preservation of intermediate graph: Some pair case]
\label{lemma:FewToAllShortcuts-Length-on-v0S}
Let $G = (V, E)$ be a weighted graph of tree-width $p$, $T$ be its corresponding tree-decomposition and $w: E \to \R^+$ be the weight function. Let $V_0 \subseteq V$ be an arbitrary subset of vertices. Let $(G^\prime, w^\prime)$ be the output of $\textproc{ConstructGraph}(G, w, T, V_0)$ and $S$ be the bag selected on first call of \textproc{ComputeEdges} on line \ref{alg:line:use-s} in Algorithm \ref{alg:compute-shortcut-edges}.
Then, for all $v \in V_0 \cup S$ and $u \in V$, we have $d_{G, w}(v, u) = d_{G^\prime, w^\prime}^{l}(v, u)$ where $l=\max(2, \log_{1.5} \vert V \vert)$.
\end{lem}

Now, it's time to improve upon the result of the last lemma. In the next lemma, we will prove that the result is also correct for all pairs shortest distances.

\begin{lem}[Distance preservation of intermediate graph: All pair case]
\label{lemma:Gprime-and-G-APSD-is-same}
Let $G = (V, E)$ be a weighted graph of tree-width $p$ and the weight function $w: E \to \R^+$. Let $T$ be a tree-decomposition of $G$, and let $(G^\prime, w^\prime) = \textproc{ConstructGraph}(G, w, T, \emptyset)$. Let $S$ be the bag selected on the first call of \textproc{ComputeEdges} on line \ref{alg:line:use-s} in Algorithm \ref{alg:compute-shortcut-edges}.
Then, for all $v$ and $u \in V$, we have $d_{G, w}(v, u) = d_{G^\prime, w^\prime}^{l}(v, u)$ where $l = 2\cdot \max(2, \log_{1.5} \vert V \vert)$.
\end{lem}


\begin{algorithm}[t!]
\caption{Differentially Private All-Pairs Shortest Path Distances for Low Tree-Width Graphs}
\label{alg:apsd-low-tree-width-main}
\begin{algorithmic}[1]
    \Require{Graph $G$ with weight function $w$, tree decomposition $T$ of $G$ with maximum bag size $p+1$ and a constant number $c$ and privacy parameter $\varepsilon$}
    \Ensure{An estimate of all-pairs shortest distances with low error}
    \State \Comment{First stage: construct the shortcut graph}
    \State $(G^\prime, w^\prime) \gets \Call{ConstructGraph}{G, w, T, \emptyset}$
    
    \State \Comment{Second stage: add Laplace noise to the edge weights to provide differential privacy}
    \ForAll{$e \in E(G^\prime)$}
        \State $X_e \sim \mathrm{Lap}(c \cdot (p+1)^2 \log(\vert V \vert)/\varepsilon)$
        \State $w^\prime(e) \gets w^\prime(e) + X_e$
    \EndFor
    
    \State \Comment{Third stage: use post-processing to return all-pairs shortest path distances}
    \State \Return All-pairs shortest at most $c \cdot \log\vertt{V}$-hop distances for $G^\prime$ with weight function $w^\prime$
\end{algorithmic}
\end{algorithm}

\subsubsection*{Proof of Theorem \ref{theorem:main}}
Let $\mathcal{M}(G, T, p, w, c, \varepsilon)$ be the following mechanism. Run Algorithm \ref{alg:apsd-low-tree-width-main}.
Lemma \ref{lemma:sensitiviy-of-algorithm-is-low} asserts that the $\ell_1$ sensitivity of the algorithm after the first stage is $O\left((0 + p) \cdot p \cdot \log^2 |V| + p^2\right)$. Then, there exists some $c_1$ such that the $\ell_1$ sensitivity after the first stage is at most $c_1 \cdot \left(p^2 \cdot \log^2(|V|)\right)$. This allows us to use the Laplace mechanism. By the Laplace mechanism \ref{Proposition:laplace-is-e-private}, after the second stage, the algorithm is $\varepsilon$-DP for all $c \geq c_1$. So, by post-processing \ref{Post-Processing-proposition}, $\mathcal{M}(G, T, p, w, c, \varepsilon)$ is $\varepsilon$-DP for all $c \geq c_1$.

For all $\gamma \in (0, 1)$, from the Laplace distribution, we know that for all $e \in E(G^\prime)$, $\Pr(\vertt{X_e} \geq \log(\vertt{E(G^\prime)}/\gamma) \cdot c \cdot p^2 \cdot \log^2(|V|)/\varepsilon) = \gamma/\vertt{E(G^\prime)}$, and we also know that $\vertt{E(G^\prime)} \leq \vert V \vert^2$. So by the union bound, with probability $1-\gamma$, no $X_e$ has magnitude greater than $2c \cdot \log(\vertt{V}/\gamma) \cdot p^2 \cdot \log^2(|V|)/\varepsilon$. In the last stage, we return at most $c \log \vert V \vert$-hop shortest distances. These paths contain at most $c \log \vert V \vert$ edges and so with probability $1-\gamma$, the magnitude of error in each released distance is at most $2c^2 \cdot log(\vertt{V}/\gamma) p^2 log^3(|V|)/\varepsilon$ between the returned value and the result if all $X_e$s are equal to zero. From lemma \ref{lemma:Gprime-and-G-APSD-is-same}, for $c > 2/\log(1.5)$, the result when $X_e$s are equal to zero is APSD for $G$ with weight $w$. Thus, for $c > max(c_1, 2/\log(1.5))$, $\mathcal{M}(G, T, p, w, c, \varepsilon)$ is an $\varepsilon$-DP mechanism and with probability $1-\gamma$, has an error less than $2c^2 log(\vertt{V}/\gamma) p^2 log^3(|V|)/\varepsilon$, which is $O(log(1/\gamma) p^2 log^4(|V|)/\varepsilon)$.

Note that in the proof of Theorem \ref{theorem:main}, we took the mechanism $\mathcal{M}(G, T, p, w, c, \varepsilon)$ as Algorithm \ref{alg:apsd-low-tree-width-main} on the input graph $G$, weight function $w$, tree decomposition $T$, constant $c$, and privacy parameter $\varepsilon$, and returns its result.

\subsection{Time Complexity}
\label{section:complexity}
In this sub-section, we analyze the time complexity of Algorithm \ref{alg:apsd-low-tree-width-main} and show that it has a polynomial-time implementation. Lemma \ref{lemma:construct-intermediate-graph-compleixy} shows that the first stage of the algorithm is polynomial-time and Theorem \ref{theorem:complexity} shows that our main algorithm is a polynomial-time algorithm.

\begin{lem}
\label{lemma:construct-intermediate-graph-compleixy}
The time complexity of Algorithm \ref{alg:cunstruct-shortcut-graph} on an input graph $G=(V, E)$ is $O(\vert V \vert^3 \log \vert V \vert)$.
\end{lem}

\begin{onlyarxiv}
The proof of Lemma \ref{lemma:construct-intermediate-graph-compleixy} is deferred to the appendix.
\end{onlyarxiv}
\begin{onlyisncc}
The proof of lemma can be found in appendix of \cite{arxivversion}.
\end{onlyisncc}

\begin{theo}
\label{theorem:complexity}
The time complexity of Algorithm \ref{alg:apsd-low-tree-width-main} on input graph $G=(V, E)$ is $O(\vert V \vert^3 \log \vert V \vert)$.
\end{theo}

\begin{proof}
Lemma \ref{lemma:construct-intermediate-graph-compleixy} shows that the first stage of the algorithm has $O(\vert V \vert^3 \log \vert V \vert)$ time complexity. In the second stage, we run a for loop on every edge of the intermediate graph, so its time complexity is $O(\vert V \vert^2)$. The third stage of the algorithm can be implemented using dynamic programming. Let $d(v, u, k)$ be the $k$-hop shortest distance on $G^\prime$ between $v$ and $u$. We can fill it if $k = 0$ or $k = 1$ in $O(\vert V \vert^2)$ time. Then, for every $v, u \in V, 1 \leq k \leq c \cdot \log \vert V \vert$, $d(v, u, k)$ is equal to the minimum of $d(v, z, k-1) + w^\prime(z, u)$ for every $z \in V$. So we can compute it in $O(\vert V \vert^3 \log \vert V \vert \cdot c)$ time. We assume $c$ is constant, so the running time of the algorithm is $O(\vert V \vert^3 \log \vert V \vert)$.
\end{proof}


    




\section{Conclusion}
We proposed a polynomial time differentially private algorithm for the problem of all pair shortest distance problem for the class of low tree-width graphs. This class of graphs contains trees as a subclass. Thus, our founding generalizes the result of Sealfon in \cite{Shortest-Paths-and-Distances-with-Differential-Privacy-Sealfon}. Despite the fact that general low tree-width graphs are fundamentally more complex than trees, we manage to achieve essentially the same additive error term as that of the work of Sealfon for trees. 

On the other hand, while the algorithm of the work of \cite{chen2023differentially} provides a differentially private algorithm for the all pair shortest distance problem for general graphs, when restrict to the class of low tree-width
graphs (i.e tree-width of order $o(n^{1/3})$), our algorithm benefits from significantly lower additive error. i.e. $O(p^2 \operatorname{polylog}(n) /\varepsilon)$ versus 
$O(n^{2/3} \operatorname{polylog}(n)/\varepsilon)$.

\begin{onlyisncc}
\section*{Acknowledgment}
This research was in part supported by a grant from IPM (No. 1402050113)

\end{onlyisncc}

\bibliographystyle{IEEEtranS}

\bibliography{IEEEexample}

\begin{onlyarxiv}

\section{Appendix}
\label{section:appendix}

\subsection{Proof of Lemma \ref{lemma:DepthOfAlgorithmisNotBig}}
Let $d(n)$ denote the maximum recursion depth of the \textproc{ComputeEdges} function on graph inputs with $n$ vertices with tree-width $p$, we have:
\[
d(n) \leq 
  \left\{ \begin{array}{ll}
    1  & n \leq 6 \cdot (p+1) \\
    d(n/2 + p + 1) + 1  & Otherwise \\
  \end{array}\right.
\]
Thus, we claim that for $C \geq \frac{1}{\log 1.5}$, $d(n) \leq C\log n$
we show the result by induction.

For $n \leq 6 (p+1)$, we have $d(n) = 1 = O(\log n)$.
Otherwise, we have:
\begin{align*}
	\operatorname{d}(n) &= 1 + d(n/2 + p + 1) \\
	&\leq 1 + C(\log (n/2+p+1)) \\
	&\leq 1 + C(\log (n / c)) & (3n / 2 > p + 1 + n/2) \\
	&\leq 1  - C\log 1.5 + C\log (n) \\
	& \leq C \log (n) & (C \geq 1/\log 1.5)
	\end{align*}
	Hence, 
	$d(n) = O(\log n)$ and by induction, the statement is proved.

\subsection{Proof of Lemma \ref{lemma:SizeOfStartingSetsAreNotBig}}
During the execution of the algorithm, the size of the starting vertex set increases by at most $p+1$ in each recursive step. According to Lemma \ref{lemma:DepthOfAlgorithmisNotBig}, the maximum recursion depth is $O(\log \vert V \vert)$. Therefore, the size of the starting vertex set is bounded by $O(\vert V_0 \vert + p \cdot \log \vert V \vert)$.

\subsection{Proof of Lemma \ref{lemma:sensitiviy-of-algorithm-is-low}}
During each recursive call of the \textproc{ConstructGraph} function (except for the base case), with a starting set of size $k$, $(k + p) \cdot p$ edges are added to the returned list. By Lemma \ref{lemma:SizeOfStartingSetsAreNotBig}, $k = O(\vert V_0 \vert + p \cdot \log \vert V \vert)$. For each edge, we add two public vertices that are calculated without using private data and a shortest distance. Since the $\ell_1$ sensitivity of each shortest distance is 1, returning all $(k+p)\cdot p$ edges has an $\ell_1$ sensitivity of $O(p \cdot (\vert V_0 \vert + p) + p^2 \cdot \log \vert V \vert)$, which is $O((\vert V_0 \vert + p) \cdot p \cdot \log \vert V \vert)$.

In each recursive call, we call some graphs for the next level which are edge-disjoint. So at each recursion depth level, the graphs in each call are edge-disjoint from one another, so the $\ell_1$ sensitivity of the result at each recursion depth remains $O((\vert V_0 \vert + p) \cdot p \cdot \log \vert V \vert)$. At the last recursion depth, the number of vertices in the graph is less than $6p^2$ and we return all-pairs shortest distances, so the $\ell_1$ sensitivity at this depth level is $O(p^2)$. On the other hand, by Lemma $\ref{lemma:DepthOfAlgorithmisNotBig}$, there are at most $O(\log \vert V \vert)$ recursion depth levels. Therefore, the $\ell_1$ sensitivity of $f$ is $O((\vert V_0 \vert + p) \cdot p \cdot \log^2 \vert V \vert + p^2)$.

\subsection{Proof of Lemma \ref{lemma:FewToAllShortcuts-Length-on-v0S}}
In algorithm \ref{alg:cunstruct-shortcut-graph}, the weights of all added edges are the weights of some path in $G$, so $d_{G, w}(v, u) \leq d_{G^\prime, w^\prime}^{l}(v, u)$ for all $v, u \in V$ and $l \in [\vert V \vert]$.
\\
We will use induction on $\vertt{V}$ to prove we have:
\begin{equation}
\label{equation:lemma-length-v0-to-all-proof}
d_{G^\prime, w^\prime}^{t}(v, u) \leq d_{G, w}(v, u)
\end{equation}
for $t=\max(2, \log_{1.5} \vert V \vert)$ and for all $v \in V_0 \cup S$ and $u \in V$. If $\vertt{V} \leq 6(p+1)$, 
then, $w^\prime(v, u) = d_w(v, u)$ which concludes the result for this case.

Now assume that the claim holds for graphs of tree-width at most $p$ and fewer vertices than $n$. Consider a weighted graph $G = (V, E)$ with weight function $w: E \to \R^+$ and 
$\vertt{V} = n$.
For any $v \in V_0 \cup S$ and $u \in V$, let P be $v = v_0, v_1, \ldots, v_t = u$ shortest path in $G$. 
If $u \in S$, we have $d_{G^\prime, w^\prime}^1(v, u) = w(P)$ which makes the result trivial. Otherwise, if $u \in G \setminus S$, let $C_1, \ldots, C_l$ be the connected components of $G \setminus S$ and let $U_1, \ldots U_t$ be the subgraphs that are constructed on line \ref{alg:line:def-ui} using the sets $C_i$. Also define the sets $V_i$, weight functions $w_i$, and tree decompositions $T_i$ following the algorithm. Thus there exists an index $i \in [l]$ such that $u \in C_i$.
Let graph $G_i$ with weight function $w_i$ be the return of $\textproc{ConstructGraph}(U_i, w_i, T_i, V_i)$.

If the path $P$ is completely contained in $U_i$, according to induction hypothesis and since $\vertt{V(U_i)} \leq n$, \ref{equation:lemma-length-v0-to-all-proof} holds for $G_i$ and weight $w_i$. Since this graph a subgraph of $G^\prime$ and all edge weights of $G^\prime$ is less than edge weights of $G_i$, we can conclude the result in this case. Otherwise, if $P$ is not in fully on $U_i$, then, since $C_i \subseteq U_i$, $P$ is not completely contained in $C_i$. Because
$\{C_1, \ldots, C_t\}$ are the connected components of 
$G \setminus S$, thus there exists an index $j \in [t]$ such that the vertex $v_j$ lies in $S$. Without loss of generality we can assume that $j$ is the largest index such that vertex $v_j$ lies in $S$. Then, the subpath $P^\prime$ consisting of vertices $v_j, v_{j+1}, \ldots, v_t$ lies entirely within $U_i$. The induction hypothesis implies that 
$d_{G_i, w_i}^{\max(2, \log_{1.5} \vert V_i \vert)}(v_j, v_t) \leq w(P^\prime)$.

Because 
$n > 6(p+1)$, we have $\vert{V_i} \leq n / 1.5$. Therefore, we have 
$\log_{1.5} \vertt{V_i} \leq \log_{1.5} \vertt{V} - 1.$

For the last step note that because all edge weights in $G^\prime$ is less than or equal to those in $G_i$ and $G_i$ is a subgraph of $G^\prime$, we have 
$d_{G^\prime, w^\prime}^{\log_{1.5}(\vert V \vert) - 1}(v_j, v_t) \leq w(P^\prime)$.
Since $v_j \in S$, there is an edge between $v_0$ and $v_j$ in $G^\prime$ with weight $d_{G, w}(v_0, v_j) = w(v_0, v_1, \ldots v_j)$. 
Thus for $l=\log_{1.5}(\vert V \vert)$, we have $d_{G^\prime, w^\prime}^{l}(v_0, v_k) \leq w(v_0, v_1, \ldots v_j) + w(v_j, v_{j+1}, \ldots v_t) = w(P)$. Because $n > 6(p+1)$, thus $\log_{1.5}(\vert V \vert) > 2$, so the claim is true and thus, the lemma is true.

\subsection{Proof of Lemma \ref{lemma:Gprime-and-G-APSD-is-same}}
We prove the lemma by induction. If $\vert V \vert \leq 6(p+1)$, the result follows trivially. Now assume that the result holds true for all graphs of tree-width at most $p$ and less than $n$ vertices. Let $G$ be a graph with $n$ vertices with weight $w: E \to \R^+$ and tree-width at most $p$ and $T$ its corresponding tree decomposition. Let $v, u \in V$ be two arbitrary vertices and let path $P = v = v_0, v_1, \ldots v_t = u$ be the shortest path between $v$ and $u$ in $G$ with weight function $w$. If there exists some $i \in \{0, 1, \ldots, t\}$ such that $v_i \in S$, then, Lemma \ref{lemma:FewToAllShortcuts-Length-on-v0S} implies that $d_{G, w}(v, v_i) = d_{G^\prime, w^\prime}^{\max(2, \log_{1.5} \vert V \vert)}(v, v_i)$ and 
$d_{G, w}(v_i, u) = d_{G^\prime, w^\prime}^{\max(2, \log_{1.5} \vert V \vert)}(v_i, u)$. Since $v_i$ is on the shortest path between $v$ and $u$, we have $d_{G, w}(v, u) = d_{G, w}(v, v_i) + d_{G, w}(v_i, u)$. We also can deduce that $d_{G^\prime, w^\prime}^{2 \cdot \max(2, \log_{1.5} \vert V \vert)}(v, u) \leq d_{G^\prime, w^\prime}^{\max(2, \log_{1.5} \vert V \vert)}(v, v_i) + d_{G^\prime, w^\prime}^{\max(2, \log_{1.5} \vert V \vert)}(v_i, u)$ which implies that $d_{G^\prime, w^\prime}^{2 \cdot \max(2, \log_{1.5} \vert V \vert)}(v, u) \leq d_{G, w}(v, u)$. In proof of Lemma \ref{lemma:FewToAllShortcuts-Length-on-v0S}, we also showed that $d_{G^\prime, w^\prime}^{2 \cdot \max(2, \log_{1.5} \vert V \vert)}(v, u) \geq d_{G, w}(v, u)$ so the result implies in this case. Otherwise assume that there is no $i$ such that $v_i \in S$. Then, path $P$ is completely contained on one of connected components of $G \setminus S$ which implies the result by induction hypothesis.

\subsection{Proof of Lemma \ref{lemma:construct-intermediate-graph-compleixy}}
The algorithm consists of two steps. At first, it runs the $\textproc{ComputeEdges}(\cdot)$ function and then constructs the graph. The second step runs a for loop on every item of the result of the first step, so if the time complexity of the first stage be $T(\vert V \vert)$, the time complexity of the second stage is $O(T(\vert V \vert) + \vert V \vert^2)$, $O(\vert V \vert^2)$ for initializing $w^\prime$ and $O(T(\vert V \vert)$ for iterating $R$. Now, we prove that the $T(\vert V \vert) = O(\vert V^3 \vert \log \vert V \vert)$. The \textproc{ComputeEdges} function is a recursive function. Lemma \ref{lemma:DepthOfAlgorithmisNotBig} shows that the depth of it is $O(\log \vert V \vert)$. The recursion call of each instance of the function is based on connected components created when removing a bag from input, thus, in each depth-level of recursion call, at most $\vert V \vert$ instances of the function can be executed simultaneously. If $\vert V \vert \leq 6(p+1)$, the running time of function is $O(\vert V \vert^2)$. Otherwise, computing $S$ and constructing $H_i$s, $V_i^\prime$s and $T_i$s also has $O(\vert V \vert^2)$ time complexity. Thus, the time complexity of each instance of \textproc{ComputeEdges} without recursion calls is $O(\vert V\vert^2)$.

Thus, the time complexity of \textproc{ComputeEdges} on input graph $G=(V, E)$ is $O(\vert V \vert^2 \cdot \vert V \vert \cdot \log \vert V \vert) = O(\vert V \vert^3 \log \vert V \vert)$ and the lemma is proved.

\end{onlyarxiv}

\end{document}